\documentclass{nextgame2026} 

\usepackage{mathtools}          
\usepackage{mathrsfs}           
\mathtoolsset{showonlyrefs}     
\usepackage{graphicx}           
\usepackage{subcaption}         
\usepackage[space]{grffile}     
\usepackage{url}
\usepackage{mathrsfs}          
\usepackage{hyperref}
\usepackage{enumitem}

\usepackage{mathtools}
\usepackage{booktabs}
\usepackage{multirow}
\usepackage{zref-clever}

\DeclareMathOperator*{\argmax}{arg\,max}

\title[Parametric Open Source Games]{Parametric Open Source Games}

\optauthor{%
\Name{Aleksandar Todorov\textsuperscript{$\star$}} \Email{a.todorov.4@student.rug.nl}\\
\Name{Jesse ten Napel} \Email{j.ten.napel@student.rug.nl}\\
\Name{Alexander Müller} \Email{a.k.muller.2@student.rug.nl}\\
\addr University of Groningen, Groningen, the Netherlands \\
\addr Safe AI Netherlands}

\begin{document}

\maketitle

\begin{abstract}%
Open-source game theory studies agents whose behavior may depend on one another's decision procedures, but most existing models use discrete or symbolic programs. We introduce \emph{parametric open-source games}, a continuous analogue of program equilibria in which players choose parameter vectors and semantics maps convert the full parameter profile into mixed actions in an underlying finite game. We establish equilibrium existence results, derive an exact coupling threshold at which selfish gradient ascent in symmetric $2\times2$ games switches from defection toward cooperation, and give a one-dimensional boundary test for parametric program Nash equilibria. We further extend the framework to a neural semantics class whose first-order cooperation condition is governed by the ratio of cross-player to self-player sensitivity. Across canonical games, the framework shows how access to internal parameterizations can qualitatively reshape learning dynamics and equilibrium structure, and how sufficiently strong open-source coupling can steer selfish optimization toward cooperative outcomes.
\end{abstract}


\section{Introduction}
\label{introduction}

Open-source game theory studies settings where agents are able to inspect one another's internal decision procedures. In the classical program equilibrium framework, players submit programs that may inspect the opponent's program before outputting an action in an underlying base game \citep{tennenholtz_program_2004, fortnow_program_2009}. This possibility can qualitatively change the strategic reasoning of agents and consequently cooperation, punishment, and coordination may become contingent on what an agent proves, simulates, or infers about its opponent’s internal policy, yielding outcomes that differ sharply from the Nash equilibria of the underlying normal form game \citep{lavictoire2014program,critch_parametric_2019,critch_cooperative_2022}. Most existing open-source models are nevertheless symbolic or proof-theoretic \citep{oesterheld_robust_2019, barasz_robust_2021}, while modern learning systems are usually continuous, differentiable, and parameterized by high-dimensional vectors \citep{mnih_human-level_2015,vaswani_attention_2017,silver_mastering_2017}. Conversely, differentiable games and multi-agent optimization study smooth parameterized objectives, but usually treat each player's policy as depending only on its own parameters \citep{fudenberg_theory_1998,basar_dynamic_1999,chasnov_convergence_2020, mazumdar_gradient-based_2020, lin_finite-time_2020}.

In this work, we develop a parametric model that preserves the central open-source feature, namely that a player's behavior may depend on an opponent's internal description, while enabling standard tools from continuous optimization. In a \textit{parametric open-source game}, each player chooses a parameter vector, and a continuous semantics map converts the full parameter profile into mixed actions in an underlying finite game. The induced game is ordinary and continuous over parameter space, but it retains the open-source feature that behavior can depend on opponents' internal descriptions, providing a connection between program-style transparency and gradient-based learning.

Our contributions are threefold. We formalize parametric open-source games and give basic equilibrium existence results for the induced parametric game. We then analyze a sigmoid semantics in symmetric $2\times2$ games, deriving a first-order cooperation threshold and a boundary PPNE test, both supported by empirical verification. Finally, we introduce a neural semantics class that preserves the same first-order criterion through the ratio of cross-player to self-player sensitivity. The framework suggests that open-source strategic reasoning need not be tied to symbolic representations, as once encoded by a continuous semantics map, it can reshape both learning dynamics and equilibrium structure in ways that are analyzable with standard tools from continuous optimization.

\section{Parametric Open-Source Games}
\label{sec:posgs}

Proofs of formal statements are deferred to \zcref[S]{app:proofs}. Let $G=(N,(S_i)_{i\in N},(u_i)_{i\in N})$ be a finite normal form game, where $N = \{1, \dots, n \}$ denotes a set of $n$ agents with each agent $i \in N$ having a finite action set $S_i$ and utility $u_i:S\to\mathbb{R}$ with $S=\prod_i S_i$. For a mixed profile $\sigma\in\prod_i\Delta(S_i)$, we write the (continuous) multilinear extension of $u_i$ as $u_i(\sigma)=\sum_{s\in S}\Bigl(\prod_{j\in N}\sigma_j(s_j)\Bigr)u_i(s)$.
    
Each player $i$ chooses a parameter vector in a nonempty compact convex set $\Theta_i\subset\mathbb{R}^{d_i}$, and we denote by $\Theta=\prod_i\Theta_i$ and $\theta=(\theta_i)_{i\in N}$. Next, we state the two central definitions of our model.

\begin{definition}[Parametric open-source game]
A semantics for player $i$ is a continuous map \\ $\phi_i:\Theta\to\Delta(S_i)$. Given semantics $\phi=(\phi_i)_i$, the induced payoff is
\begin{equation}
    U_i(\theta)=u_i(\phi_1(\theta),\ldots,\phi_n(\theta)).
\end{equation}
The induced parametric game is then $\bar{G} = (N,(\Theta_i)_{i\in N},(U_i)_{i\in N})$. It is closed-source if every $\phi_i$ depends only on $\theta_i$, and open-source otherwise.
\end{definition}


\begin{definition}[Parametric program Nash equilibrium]
A profile $\theta^\star\in\Theta$ is a parametric program Nash equilibrium (PPNE) if, for every $i$ and every $\theta_i\in\Theta_i$,
\begin{equation}
    U_i(\theta_i^\star,\theta_{-i}^\star)\ge U_i(\theta_i,\theta_{-i}^\star).
\end{equation}
\end{definition}

Since each $U_i$ is continuous on compact $\Theta$, Glicksberg's theorem gives a mixed-strategy Nash equilibrium over $\Theta$. If $\theta_i\mapsto U_i(\theta_i,\theta_{-i})$ is quasiconcave on $\Theta_i$ for each fixed $\theta_{-i}$, Kakutani's theorem also gives a pure PPNE. We prove this in \zcref[S]{app:existence}. We next specialize the framework to two-action games, where the effect of open-source dependence can be analyzed explicitly.

\section{Semantics Families and Game Dynamics}\label{sec:semantics_families}

We now instantiate the framework with semantics maps that are simple enough to analyze explicitly while still exhibiting open-source effects. Experimental details are deferred to \zcref[S]{app:experiments}.

\subsection{Sigmoid Semantics}
\begin{figure}[thbp]
    \centering
    \includegraphics[width=\linewidth]{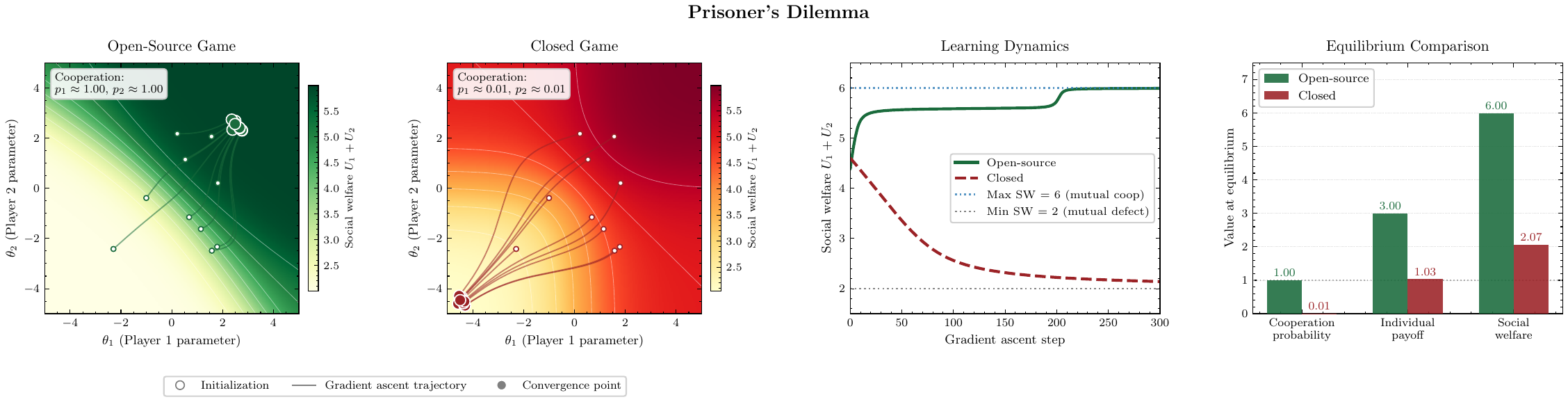}
    \caption{Open-source versus closed-source sigmoid semantics in the Prisoner's Dilemma. Under open-source semantics, projected gradient-ascent trajectories move toward the cooperative corner and high welfare, whereas closed-source dynamics move toward defection.}
    \label{fig:pd_open_vs_closed}
\end{figure}

We adopt a two-action setting $S_i=\{C,D\}$ and use a one-dimensional parameter $\theta_i\in[-B,B] \subseteq \mathbb{R}$. The sigmoid semantics determines the probability of playing $C$ (cooperation) as
\begin{equation}
    p_i(\theta)=\phi_i(\theta)(C)=\sigma(\theta_i+\gamma\theta_{-i}),
    \qquad \sigma(x)=(1+e^{-x})^{-1}.
\end{equation}
The coupling $\gamma\ge0$ measures how strongly player $i$'s behavior responds to the opponent's parameter, while $\gamma=0$ recovers the closed-source parameterization.

As a base experiment, we study selfish projected gradient ascent in the induced parametric game,
\begin{equation}
    \theta_i^{t+1}=
    \Pi_{[-B,B]}\!\left(\theta_i^t+\eta\,\partial_{\theta_i}U_i(\theta^t)\right),
    \qquad i\in\{1,2\}, \; \eta > 0,
\end{equation}
using finite-difference gradients of $U_i$ with respect to each player's own parameter. The underlying interaction remains a one-shot symmetric $2 \times 2$ game.

The trajectories in \zcref[S]{fig:pd_open_vs_closed} visualize how the semantics map changes and which regions of parameter space are attractive under this learning rule. In the Prisoner's Dilemma, closed-source dynamics move toward defection, while sufficiently coupled open-source dynamics move toward the cooperative boundary. Nonetheless, the cooperation results and basin of attraction strongly depend on the choice of $\gamma$. The following theorem provides the exact conditions under which, for a general symmetric $2 \times 2$ game with payoffs $R$ for mutual cooperation, $S$ for being exploited, $T$ for exploiting, and $P$ for mutual defection and an arbitrary $\gamma \geq 0$, the cooperative regime first emerges under gradient ascent as a local property.

\begin{theorem}[Phase transition]
\label{thm:critical_coupling}
Assume $R+T-P-S>0$ and define
\begin{equation}
    \gamma^\star=\frac{T+P-R-S}{R+T-P-S}.
\end{equation}
For initializations near $\theta=(0,0)$, selfish gradient ascent points toward lower cooperation when $\gamma<\gamma^\star$ and toward higher cooperation when $\gamma>\gamma^\star$.
\end{theorem}

Thus $\gamma^\star$ is the local open-source coupling needed to reverse the initial incentive of gradient ascent. For Stag Hunt, $\gamma^\star\leq 0$, so cooperation is already locally attractive. We verify \zcref[S]{thm:critical_coupling} empirically across four canonical games, and \zcref[S]{fig:critical_coupling_phase_transition} shows that the empirical transition in mean cooperation matches the analytical value of $\gamma^\star$ closely.

\begin{figure}[thbp]
    \centering
    \includegraphics[width=0.92\linewidth]{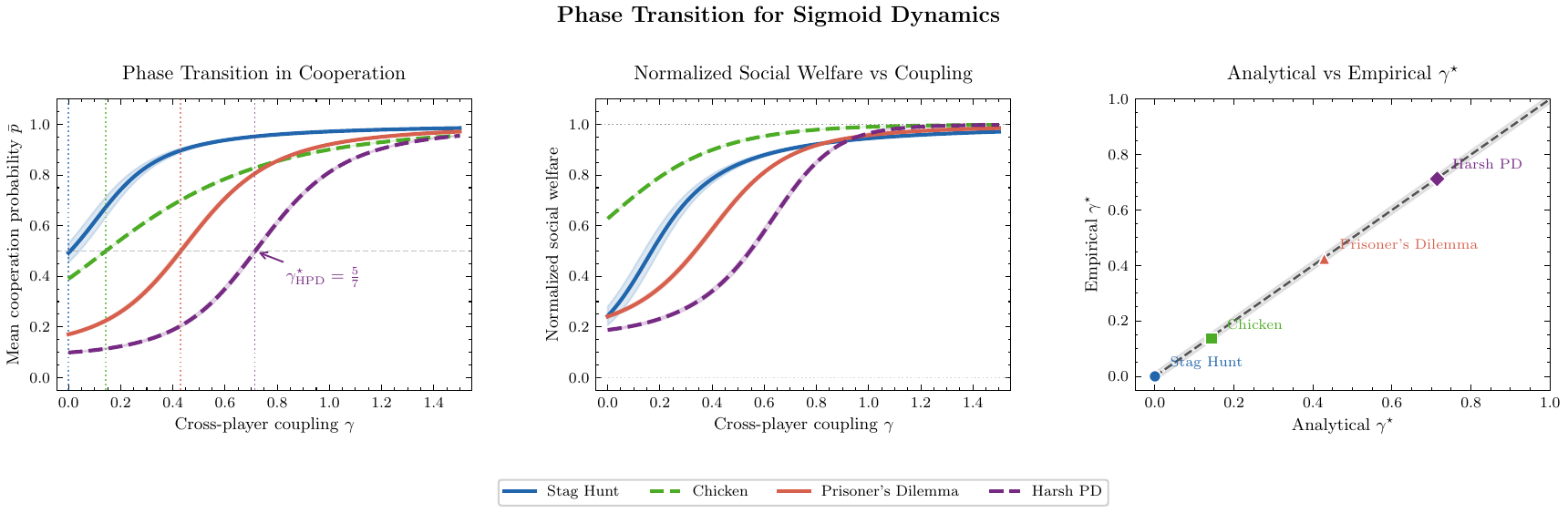}
    \caption{Phase transition in $2\times2$ games. Panel (a) shows mean terminal cooperation as a function of $\gamma$, with dotted lines marking $\gamma^\star$ and the dashed line marking $\bar p=0.5$. Panel (b) shows normalized social welfare, and Panel (c) compares analytical and empirical transition points. All results are averaged over 20 seeds.}
    \label{fig:critical_coupling_phase_transition} 
\end{figure}

\subsection{Boundary Equilibria}
The result from \zcref[S]{thm:critical_coupling} describes only the local direction of gradient ascent near the symmetric midpoint. It therefore explains when learning initially moves toward cooperation, but not whether the limiting boundary point is stable against unilateral deviations. This distinction matters because projected gradient ascent in the sigmoid model often converges to the boundary of $\Theta=[-B,B]^2$. Hence, the natural equilibrium candidates are boundary points such as $(B,B)$ and $(-B,-B)$.

\begin{theorem}[Boundary PPNE characterization]
\label{thm:boundary_characterization}
For fixed $\gamma\geq0$, suppose that $\theta_i\mapsto U_i(\theta_i,B)$ is differentiable on $(-B,B)$ and satisfies
\begin{equation}
    \frac{d}{d\theta_i}U_i(\theta_i,B)\geq0
    \qquad
    \text{for all }\theta_i\in(-B,B)
    \text{ and } i\in\{1,2\}.
\end{equation}
Then $(B,B)$ is a PPNE.
\end{theorem}

\zcref[S]{thm:boundary_characterization} shows that equilibrium verification does not require analyzing the full two-dimensional landscape. For each player, we freeze the opponent at the candidate boundary value and maximize over the player's own parameter interval. \zcref[S]{fig:ppne_verification} applies this test in the Prisoner's Dilemma. Under open-source coupling, the unilateral payoff is maximized at the cooperative boundary; under closed-source semantics, it is maximized at the defective boundary. This verifies the limit points observed in \zcref[S]{fig:pd_open_vs_closed} as PPNEs. A weaker-coupling example in which open-source semantics does not yield a PPNE is given in \zcref[S]{app:exp_3_ppne_verification} and \zcref[S]{fig:ppne_verification_nonppne}.

\begin{figure}[thpb]
    \centering
    \includegraphics[width=0.92\linewidth]{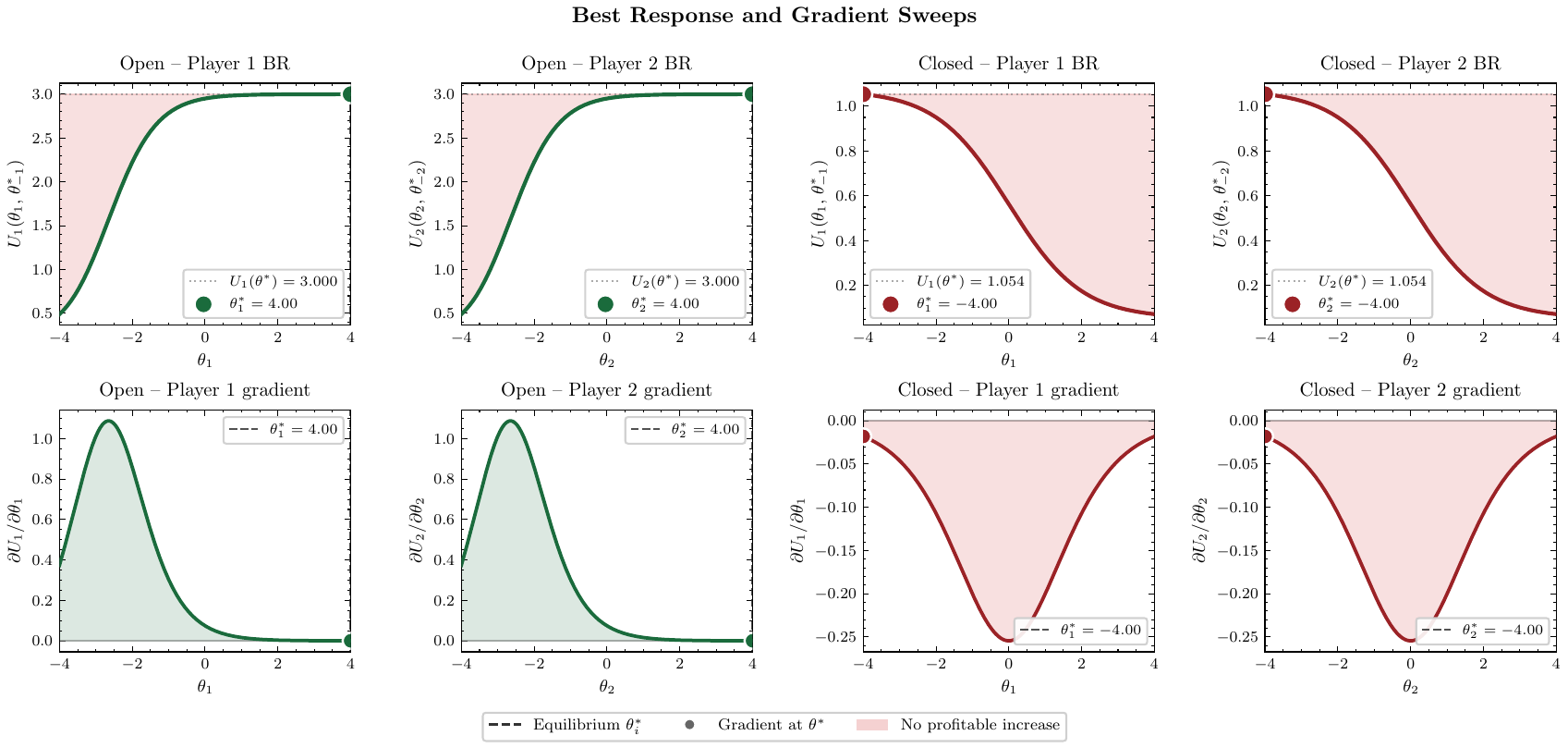}
    \caption{Boundary PPNE verification in the Prisoner's Dilemma. For each candidate equilibrium, the opponent is fixed at its boundary value, and the remaining player's payoff is evaluated over $\theta_i\in[-B,B]$. Open-source semantics maximizes payoff at the cooperative boundary, while closed-source semantics maximizes payoff at the defective boundary.}
    \label{fig:ppne_verification}
\end{figure}

\subsection{Neural Open-Source Semantics}\label{sec:neural_semantics}
The sigmoid semantics isolates open-source dependence through a single scalar $\gamma$, which makes the threshold analysis transparent. However, it also hard-codes a linear dependence on the opponent's parameter. To allow richer semantics while preserving the same local interpretation, we consider
\begin{equation}
    p_i(\theta)
    =
    \sigma\!\left(
    \alpha\theta_i+\beta\theta_{-i}
    +h(\theta_i^2,\theta_{-i}^2,\theta_i\theta_{-i};W)
    \right), \qquad h(x;W)=W_2^\top\tanh(W_1^\top x+b_1),
\end{equation}
where $K \in \mathbb{N}$ is the network hidden width and $W = (W_1, b_1, W_2)$, where $W_1 \in \mathbb{R}^{3 \times K}$, $b_1 \in \mathbb{R}^K$, $W_2 \in \mathbb{R}^K$. The (learnable) scalars $\alpha \in \mathbb{R}$ and $\beta \in \mathbb{R}$ have a direct interpretation, namely that $\alpha$ measures first-order sensitivity to the player's own parameter, while $\beta$ measures first-order sensitivity to the opponent's parameter. The residual network $h$ adds nonlinear dependence on the full parameter profile, but it receives only quadratic features. Therefore, its first-order derivative vanishes at $\theta=(0,0)$, so it cannot change the local incentive condition at the midpoint. This constraint ensures the model is more expressive away from the midpoint, but its first-order open-source coupling remains identifiable.

Under this constraint, the sigmoid threshold from \zcref[S]{thm:critical_coupling} extends cleanly. If $\alpha>0$, then near $\theta=(0,0)$ the relevant coupling is no longer $\gamma$ itself, but the sensitivity ratio $\beta/\alpha$. In particular, selfish gradient ascent initially points toward cooperation exactly when $\frac{\beta}{\alpha}>\gamma^\star$ with $\gamma^\star$ as in \zcref[S]{thm:critical_coupling}. The formal statement and proof are given in \zcref[S]{app:neural_threshold}.

\zcref[S]{fig:neural_semantics_learning} tests whether this first-order condition continues to organize learning once the semantics are neural. We compare sigmoid closed-source semantics, sigmoid open-source semantics, fixed neural semantics with $\beta=0$, fixed neural semantics with $\beta/\alpha=1.5$, and two jointly learned neural semantics. In the learned conditions, both the player parameters and the neural semantics parameters are optimized. The warm-start condition initializes $(\alpha,\beta)=(1,1.5)$, already above the cooperative threshold in the relevant games, whereas the cold-start condition initializes $(\alpha,\beta)=(0,0)$.

The results show that fixed neural semantics reproduce the sigmoid baselines when their first-order ratios match. Neural-fixed-closed tracks sigmoid-closed, and neural-fixed-open tracks sigmoid-open. This supports the interpretation of $\beta/\alpha$ as the effective local open-source coupling. The learned conditions show a stronger point: warm-started neural semantics reach the same high-welfare regime as the fixed open-source models, but cold-started neural semantics do not reliably discover it. Thus, the threshold is not merely a representational condition, but also acts as an optimization barrier: the model class can represent cooperative open-source dependence, but gradient-based learning may fail to find it unless the first-order coupling is initialized in the right regime.

\begin{figure}[thpb]
    \centering
    \includegraphics[width=\linewidth]{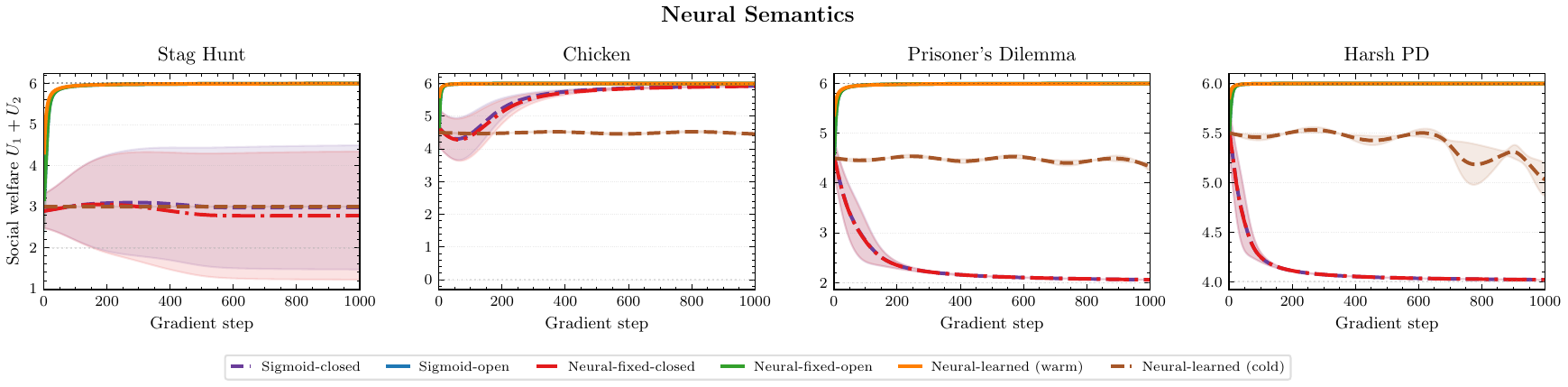}
    \caption{Learning curves under neural open-source semantics across canonical $2\times2$ games. Fixed neural semantics match the corresponding sigmoid baselines when the first-order ratio $\beta/\alpha$ is matched. Jointly learned neural semantics reach high welfare from a warm initialization above the threshold, but cold starts do not reliably discover cooperative coupling. Shaded regions indicate one standard deviation across seeds.}
    \label{fig:neural_semantics_learning}
\end{figure}

\section{Discussion and Conclusion}
\label{sec:discussion}

We introduced parametric open-source games as a continuous model of strategic interaction in which each agent chooses a parameter vector and a semantics map converts the full parameter profile into mixed actions in an underlying finite game. The model preserves the central mechanism of open-source reasoning \citep{tennenholtz_program_2004}, namely that strategic behavior may be conditioned on another agent's internal description, while making the resulting game amenable to tools from continuous optimization and differentiable games.

Within this framework, our results show that open-source dependence can reshape both learning dynamics and equilibrium structure, analogous to the program equilibrium setting \citep{tennenholtz_program_2004, fortnow_program_2009, lavictoire2014program}. The sigmoidal and neural families of open-source semantics studied within this work both exhibit the same organizing principle: cooperation becomes locally attractive when cross-player sensitivity is sufficiently large relative to self-sensitivity. For the sigmoid model, this yields an explicit critical coupling, while for the neural model, the same condition is preserved through the first-order ratio between cross-player and self-player sensitivity. We then separate this local dynamical analysis from global equilibrium verification by using boundary best-response sweeps, which test whether a candidate limit point is stable against unilateral deviations.

The presented models are intentionally idealized. They assume full access to opponent parameters, focus on symmetric two-player examples, and study one-shot base games. Nevertheless, they provide an interpretable starting point for studying how transparency over internal representations can affect incentives, learning dynamics, and equilibrium structure in parametric multi-agent systems. Future work should relax these assumptions by considering partial or noisy transparency, certifiable properties instead of full parameter disclosure, asymmetric semantics, larger populations, and sequential environments such as Markov games. These extensions would connect the framework more directly to realistic multi-agent systems, where agents may observe only incomplete but strategically relevant information about one another's internal structure. Overall, parametric open-source games suggest that internal parameter access is not a superficial modeling detail, but it can change the incentives created by gradient-based learning, while introducing new robustness questions about how such access is represented, optimized, and verified.

\section*{Acknowledgement}
We thank Davide Grossi at the University of Groningen for introducing us to the field of computational game theory and for his encouragement to pursue a publication.

\bibliography{main}
\newpage
\clearpage
\appendix

\section{Proofs of Statements}\label{app:proofs}

\subsection{Existence of PPNE}\label{app:existence}

\begin{theorem}[Existence of PPNE]\label{thm:existence}
For any parametric open-source game $G$, a mixed-strategy Nash equilibrium over
$\Theta$ exists. If additionally each map $\theta_i\mapsto U_i(\theta_i,\theta_{-i})$ is
quasiconcave on $\Theta_i$ for every fixed $\theta_{-i}$, then a pure-strategy PPNE exists.
\end{theorem}

We prove the two existence claims of \zcref[S]{thm:existence} separately.

\begin{itemize}
    \item For the existence of a mixed-strategy equilibrium, observe first that each $\Theta_i$ is a nonempty compact subset of a Euclidean space, hence a compact metric space.  Moreover, by construction of the induced game, each $U_i$ is continuous on the product space $\Theta = \prod_{j\in N} \Theta_j$. Therefore, the induced parametric game is an $n$-player game with compact metric strategy spaces and continuous payoffs. By Glicksberg's generalization of Kakutani's fixed point theorem \citep{glicksberg_further_1952}, such a game admits a mixed-strategy Nash equilibrium. Hence $G$ has a mixed equilibrium over $\Theta$.
    \item We now prove the pure-strategy claim under the additional quasiconcavity assumption. For each player $i$, define the best-response correspondence
    \begin{equation}
B_i(\theta_{-i})
=
\operatorname*{arg\,max}_{\theta_i\in\Theta_i} U_i(\theta_i,\theta_{-i}).
\end{equation}
Since $\Theta_i$ is compact and $U_i(\cdot,\theta_{-i})$ is continuous for each fixed $\theta_{-i}$, the maximum is attained. Thus $B_i(\theta_{-i})$ is nonempty for every $\theta_{-i}$. Because $\Theta_i$ is compact and $U_i$ is continuous, Berge's maximum theorem \citep{berge_topological_1877} implies that $B_i$ has closed graph and is upper hemicontinuous. Since $U_i(\cdot,\theta_{-i})$ is quasiconcave on the convex set $\Theta_i$, each upper contour set $\{\theta_i\in\Theta_i : U_i(\theta_i,\theta_{-i}) \ge c\}$
is convex, and therefore the argmax set $B_i(\theta_{-i})$ is convex. Hence, for every $i$, the operator $B_i : \Theta_{-i} \rightrightarrows \Theta_i$ has nonempty, convex, compact values and is upper hemicontinuous.

Define the joint best-response correspondence as  $B(\theta) = \prod_{i\in N} B_i(\theta_{-i})$ for $\theta\in\Theta$. Since $\Theta = \prod_{i\in N}\Theta_i$ is nonempty, compact, and convex, and since each $B_i$ is upper hemicontinuous with nonempty, convex, compact values, the $B$ has the same properties. Kakutani's fixed point theorem therefore yields some $\theta^\star\in\Theta$ such that $\theta^\star \in B(\theta^\star)$. Equivalently, for every player $i$, $\theta_i^\star \in B_i(\theta_{-i}^\star),$
which means $
U_i(\theta_i^\star,\theta_{-i}^\star) \ge U_i(\theta_i,\theta_{-i}^\star)$
for all $\theta_i\in\Theta_i$.
Thus $\theta^\star$ is a pure-strategy Nash equilibrium of the induced parametric game.
\end{itemize}

\subsection{Phase Transition for Sigmoidal Dynamics}
\label{app:gradients}

We prove a slightly more general version as presented in \zcref[S]{thm:critical_coupling}. For the symmetric $2\times 2$ base game with payoffs $R,S,T,P$, player $1$'s expected payoff as a function of cooperation probabilities $p_1,p_2$ is
\begin{equation}
    U_1(p_1,p_2)
    =
    p_1p_2R + p_1(1-p_2)S + (1-p_1)p_2T + (1-p_1)(1-p_2)P.
\end{equation}
Hence
\begin{align}
    \frac{\partial U_1}{\partial p_1}
    &= p_2(R-T) + (1-p_2)(S-P), \\
    \frac{\partial U_1}{\partial p_2}
    &= p_1(R-S) + (1-p_1)(T-P).
\end{align}

Under open-source sigmoid semantics
\begin{equation}
p_i(\theta)=\sigma(\theta_i+\gamma\theta_{-i}),
\qquad
\sigma'(x)=\sigma(x)(1-\sigma(x)),
\end{equation}
we have
\begin{align}
    \frac{\partial p_1}{\partial \theta_1} &= p_1(1-p_1), \\
    \frac{\partial p_2}{\partial \theta_1} &= \gamma\,p_2(1-p_2).
\end{align}
Applying the chain rule gives
\begin{equation}
\label{eq:general_grad}
    \frac{\partial U_1}{\partial \theta_1}
    =
    \bigl[p_2(R-T)+(1-p_2)(S-P)\bigr]\,p_1(1-p_1)
    +
    \bigl[p_1(R-S)+(1-p_1)(T-P)\bigr]\,\gamma\,p_2(1-p_2).
\end{equation}

This expression separates the direct closed-source incentive from the open-source cross-player incentive. Interior critical points satisfy $\partial_{\theta_1}U_1=0$, which yields the threshold coupling
\begin{equation}
\label{eq:gamma_general}
    \gamma
    =
    \frac{p_1(1-p_1)}{p_2(1-p_2)}
    \cdot
    \frac{p_2(T-R)+(1-p_2)(P-S)}
         {p_1(R-S)+(1-p_1)(T-P)}.
\end{equation}

In the symmetric case $p_1=p_2=p$, this reduces to
\begin{equation}\label{eq:gamma_p}
    \gamma(p)
    =
    \frac{p(T-R)+(1-p)(P-S)}
         {p(R-S)+(1-p)(T-P)}.
\end{equation}

The threshold in \zcref[S]{thm:critical_coupling} is the special case corresponding to the symmetric midpoint $\theta=(0,0)$, for which $p_1=p_2=\tfrac12$. Substituting $p=\tfrac12$ into \eqref{eq:gamma_p} gives
\begin{equation}
    \gamma^\star
    =
    \gamma\left(\tfrac12\right)
    =
    \frac{T+P-R-S}{R+T-P-S}.
\end{equation}

Equivalently, evaluating \eqref{eq:general_grad} at $p_1=p_2=\tfrac12$ yields
\begin{equation}
    \frac{\partial U_1}{\partial \theta_1}\Big|_{\theta=(0,0)}
    =
    \frac18
    \Bigl[(R+S-T-P)+\gamma(R+T-S-P)\Bigr],
\end{equation}
so under the assumption $R+T-P-S>0$,
\begin{equation}
\operatorname{sgn}\left(\frac{\partial U_1}{\partial \theta_1}(0,0)\right)
=
\operatorname{sgn}\bigl(\gamma-\gamma^\star\bigr).
\end{equation}

For completeness, the limiting threshold values near full defection and full cooperation are
\begin{align}
    \gamma_{\mathrm{def}}
    &= \lim_{p\to 0}\gamma(p)
    = \frac{P-S}{T-P}, \\
    \gamma_{\mathrm{coop}}
    &= \lim_{p\to 1}\gamma(p)
    = \frac{T-R}{R-S}.
\end{align}
Moreover,
\begin{equation}
    \frac{d\gamma}{dp}
    =
    \frac{(T+S-R-P)(T-S)}
         {\bigl[p(R-S)+(1-p)(T-P)\bigr]^2},
\end{equation}
so the monotonicity of $\gamma(p)$ is determined by the sign of $T+S-R-P$. The learning dynamics claim follows from continuity of the gradient field, as near $\theta=(0,0)$, the gradient of each player's payoff with respect to their own parameter has the same sign as $\gamma-\gamma^\star$ (by symmetry), so simultaneous gradient ascent either pushes both parameters toward $+B$ or toward $-B$.

In other words, under open-source sigmoid semantics with coupling $\gamma\ge 0$, if $\gamma < \gamma^\star$ and the initialization lies in a neighborhood of $\theta = (0,0)$, each player has a local incentive to move in the direction of less cooperation, and learning dynamics starting near this midpoint drift toward defection. Conversely, if $\gamma > \gamma^\star$, each player has a local incentive to move in the direction of higher-cooperation parameter values, and learning dynamics starting near this midpoint drift toward cooperation.

\subsection{Boundary Characterizations of PPNEs}
The full statement underlying \zcref[S]{thm:boundary_characterization} is that for a fixed $\gamma\geq 0$, the following are equivalent:
\begin{enumerate}[label=(\roman*)]
    \item $(B,B)$ is a PPNE.
    \item For each $i\in\{1,2\}$, $B\in \argmax_{\theta_i\in[-B,B]} U_i(\theta_i,B)$.
    \item For each $i\in\{1,2\}$ and all $\theta_i\in[-B,B]$, $U_i(B,B)\ge U_i(\theta_i,B)$.
\end{enumerate}
If additionally $\theta_i\mapsto U_i(\theta_i,B)$ is differentiable on $(-B,B)$, then (iii) is implied by
\begin{equation}
    \frac{d}{d\theta_i}U_i(\theta_i,B)\ge 0\quad\text{for all }\theta_i\in(-B,B)\ \text{and }i=1,2.
\end{equation}

\begin{itemize}
    \item To show the equivalence claim, recall that by definition, \((B,B)\) is a PPNE if and only if, for each player \(i\in\{1,2\}\),
\[
U_i(B,B)\ge U_i(\theta_i,B)
\qquad
\text{for all }\theta_i\in[-B,B].
\]
This is exactly statement (iii). Hence (i) and (iii) are equivalent.

Next, fix \(i\in\{1,2\}\). Statement (ii) says that
\[
B\in \argmax_{\theta_i\in[-B,B]} U_i(\theta_i,B),
\]
which means precisely that
\[
U_i(B,B)\ge U_i(\theta_i,B)
\qquad
\text{for all }\theta_i\in[-B,B].
\]
This is exactly statement (iii). Thus (ii) and (iii) are also equivalent.
\item For the differentiability claim, fix $i\in\{1,2\}$ and suppose $f_i(\theta_i)=U_i(\theta_i,B)$
is differentiable on $(-B,B)$ and satisfies
$f_i'(\theta_i)\geq 0$ for all $\theta_i\in(-B,B)$.
Then $f_i$ is nondecreasing on the interval $[-B,B]$. Therefore, for every $\theta_i\in[-B,B]$,
\begin{equation}
    U_i(\theta_i,B)=f_i(\theta_i)\le f_i(B)=U_i(B,B),
\end{equation}
which shows (iii).
\end{itemize}

\subsection{Neural Semantics}
\label{app:neural_threshold}

We now prove the neural analogue of the phase transition result.

\begin{theorem}[First-order threshold for neural semantics]
\label{thm:neural_threshold}
Let $p_i$ follow the neural open-source semantics with parameters $(\alpha,\beta,W)$, and let $\alpha>0$. Define
\begin{equation}
    \gamma^\star = \frac{T+P-R-S}{R+T-P-S}.
\end{equation}
Then, at the symmetric midpoint $\theta=(0,0)$,
\begin{equation}
    \frac{\partial U_1}{\partial \theta_1}\bigg\vert_{\theta=\mathbf{0}} > 0
    \quad\Longleftrightarrow\quad
    \frac{\beta}{\alpha} > \gamma^\star.
\end{equation}
Moreover, the ratio $\beta/\alpha$ is exactly equal to the Jacobian ratio
\begin{equation}
    \frac{\partial p_i/\partial \theta_{-i}}{\partial p_i/\partial \theta_i}
\end{equation}
evaluated at $\theta=(0,0)$, independently of the residual weights $W$.
\end{theorem}

\noindent \textit{Proof.} Recall that under neural open-source semantics,
\begin{equation}
    p_i(\theta)
    =
    \sigma\left(
        \alpha \theta_i + \beta \theta_{-i}
        + h(\theta_i^2,\theta_{-i}^2,\theta_i\theta_{-i};W)
    \right),
\end{equation}
where
\[
h(\theta_i^2,\theta_{-i}^2,\theta_i\theta_{-i};W)
=
W_2^\top \tanh\left(W_1^\top[\theta_i^2,\theta_{-i}^2,\theta_i\theta_{-i}] + b_1\right).
\]
By construction, the residual term depends only on quadratic monomials in $(\theta_i,\theta_{-i})$. Hence
\begin{equation}
    \nabla_\theta h(0,0;W)=0.
\end{equation}
In particular, at the symmetric midpoint $\theta=(0,0)$, the first-order derivatives of the cooperation probabilities are determined entirely by $\alpha$ and $\beta$. Since $p_1=p_2=\sigma(0)=\tfrac12$, we have
\begin{equation}
    \frac{\partial p_1}{\partial \theta_1}(0,0)
    =
    \sigma'(0)\,\alpha
    =
    \frac{\alpha}{4},
    \qquad
    \frac{\partial p_2}{\partial \theta_1}(0,0)
    =
    \sigma'(0)\,\beta
    =
    \frac{\beta}{4}.
\end{equation}

\noindent Using the payoff derivatives from \zcref[S]{app:gradients}, evaluated at $p_1=p_2=\tfrac12$,
\begin{equation}
    \frac{\partial U_1}{\partial p_1}\Big|_{p=\frac12}
    =
    \frac{R+S-T-P}{2},
    \qquad
    \frac{\partial U_1}{\partial p_2}\Big|_{p=\frac12}
    =
    \frac{R+T-S-P}{2},
\end{equation}
the chain rule gives
\begin{equation}
    \frac{\partial U_1}{\partial \theta_1}(0,0)
    =
    \frac{\partial U_1}{\partial p_1}\frac{\partial p_1}{\partial \theta_1}
    +
    \frac{\partial U_1}{\partial p_2}\frac{\partial p_2}{\partial \theta_1}
    =
    \frac{1}{8}
    \Bigl[
        \alpha(R+S-T-P)
        +
        \beta(R+T-S-P)
    \Bigr].
\end{equation}
Assuming $\alpha>0$ and $R+T-P-S>0$, this quantity is positive if and only if
\begin{equation}
    \frac{\beta}{\alpha}
    >
    \frac{T+P-R-S}{R+T-P-S}
    =
    \gamma^\star.
\end{equation}
Thus, the local incentive at the symmetric midpoint points toward increased cooperation exactly when $\beta/\alpha>\gamma^\star$, which proves the neural threshold result. Moreover,
\begin{equation}
    \frac{\partial p_i/\partial \theta_{-i}}{\partial p_i/\partial \theta_i}\Bigg|_{\theta=(0,0)}
    =
    \frac{\beta/4}{\alpha/4}
    =
    \frac{\beta}{\alpha},
\end{equation}
so the ratio $\beta/\alpha$ is exactly the Jacobian ratio governing first-order cross-player sensitivity at the midpoint, independently of the residual weights $W$. Equivalently, $\beta/\alpha$ is the ratio between cross-player and self-player Jacobian sensitivity of the cooperation probability at the midpoint. 
\section{Experimental Details and Hyperparameters}
\label{app:experiments}
This appendix collects implementation details, optimization settings, and hyperparameters for all experiments reported in the paper. Unless otherwise stated, all experiments are run on the induced parametric game rather than as repeated play of the underlying base game. Thus, when we refer to learning dynamics, we mean numerical optimization trajectories in parameter space under projected gradient ascent. For each experiment, we report the base game, the semantics family, the optimization procedure, the initialization scheme, and the quantities visualized in the corresponding figures.

\subsection{Experiment 1: Open-Source versus Closed Semantics in the Prisoner's Dilemma}
\label{app:exp1}
This experiment compares open-source and closed-source sigmoid semantics in the one-shot Prisoner's Dilemma with payoff matrix $R = 3, P = 1, T = 5, S = 0$. The cross-player coupling for the open-source sigmoidal semantics is $\gamma=1.5$, while the closed-source baseline is recovered for $\gamma = 0$. In both cases, the parameter domain is the box $\Theta=[-B,B]^2, B=5$. For compactness, hyperparameters are also summarized in \zcref[S]{tab:exp1_hparams}.

For each initialization $\theta^0=(\theta_1^0,\theta_2^0)$, we run simultaneous projected gradient ascent on the induced payoffs,
\begin{equation}
\theta_i^{t+1}
=
\Pi_{[-B,B]}
\bigl(
\theta_i^t+\eta\,\partial_{\theta_i}U_i(\theta^t)
\bigr),
\qquad i\in\{1,2\},
\end{equation}
with learning rate $\eta=0.15$. In the implementation, the partial derivatives are approximated numerically by finite differences with step size $\varepsilon=10^{-5}$. Each trajectory is run for $10{,}000$ iterations, although only the first $300$ steps are displayed in the learning-dynamics panel of \zcref[S]{fig:pd_open_vs_closed}. Parameters are projected back to $[-5,5]$ after every update. The repeated iterations in the experiment do not correspond to repeated play of the Prisoner's Dilemma itself, but to a numerical learning procedure in the induced parametric game over $\Theta=[-B,B]^2$ over timesteps $t \geq 0$ starting from an initial parameter profile $\theta^0=(\theta_1^0,\theta_2^0)$.

Initial conditions are sampled independently and uniformly from $[-2.5,2.5]^2$,
using 10 random initializations and seed 0. To visualize the induced landscape, social welfare $U_1+U_2$ is evaluated on a $200\times 200$ grid over $\Theta$.

\begin{table}[htpb]
\centering
\caption{Hyperparameters to reproduce \zcref[S]{fig:pd_open_vs_closed}.}
\label{tab:exp1_hparams}
\begin{tabular}{ll}
\toprule
Parameter & Value \\
\midrule
Base game payoffs & $R=3,\; P=1,\; T=5,\; S=0$ \\
Open-source coupling & $\gamma = 1.5$ \\
Closed-source coupling & $\gamma = 0$ \\
Learning rate & $\eta = 0.15$ \\
Number of gradient steps & $600$ \\
Projection bound & $B = 5$ \\
Parameter domain & $\Theta=[-5,5]^2$ \\
Number of random initializations & $10$ \\
Initialization distribution & uniform on $[-2.5,2.5]^2$ \\
Random seed & $0$ \\
Finite-difference step size & $\varepsilon = 10^{-5}$ \\
\bottomrule
\end{tabular}
\end{table}

\subsection{Experiment 2: Gamma Sweep and Phase Transition}
\label{app:exp2}

This experiment largely mirrors the setup described in \zcref[S]{app:exp1}. For each game $G$, we evaluate the open-source semantics over a uniform grid of $100$ values of $\gamma$ between $0$ and $1.5$. The parameter domain is $\Theta=[-B,B]^2, B=5$. For each game and each value of $\gamma$, we run simultaneous projected gradient ascent on the induced parametric game, with learning rate $\eta=0.01$ for $600$ steps. In the implementation, the partial derivatives are approximated numerically by finite differences. For every $\gamma$, we average over $20$ random initializations. Initial parameters are sampled from a centered Gaussian distribution, with standard deviation $0.1$ for all games except Chicken, where the standard deviation is reduced to $0.01$.

For each run, we record the terminal mean cooperation probability $\bar p=\frac{p_1+p_2}{2}$and the terminal social welfare $U_1+U_2$. Across random initializations, we report the mean and standard deviation of both quantities as functions of $\gamma$. The empirical critical coupling $\gamma^\ast_{\mathrm{emp}}$ is defined as the first sampled value of $\gamma$ at which the averaged cooperation curve crosses the threshold $\bar p=0.5$. This is then compared against the analytical threshold $\gamma^\star=\frac{T+P-R-S}{R+T-P-S}$ derived in \zcref[S]{thm:critical_coupling}. Hyperparameters are summarized in \zcref[S]{tab:exp2_hparams}.

\begin{table}[htpb]
\centering
\caption{Hyperparameters to reproduce \zcref[S]{fig:critical_coupling_phase_transition}.}
\label{tab:exp2_hparams}
\begin{tabular}{ll}
\toprule
Parameter & Value \\
\midrule
Swept coupling range & $\gamma \in [0,1.5]$ \\
Number of $\gamma$ values & $100$ \\
Projection bound & $B = 5$ \\
Parameter domain & $\Theta=[-5,5]^2$ \\
Learning rate & $\eta = 0.01$ \\
Number of gradient steps & $600$ \\
Number of random initializations & $20$ \\
Initialization distribution & Gaussian, centered at $0$ \\
Initialization scale & $0.1$ for all games, $0.01$ for Chicken \\
Empirical threshold criterion & first $\gamma$ with mean $\bar p > 0.5$ \\
Finite-difference step size & $\varepsilon = 10^{-5}$ \\
Stag hunt & $R=3.0,\; S=0.0,\; T=2.0,\; P=1.0; \; \gamma^\star = 0$\\
Chicken & $R=3.0,\; S=1.0,\; T=5.0,\; P=0.0; \; \gamma^\star = 1/7$ \\
Prisoner's dilemma & $R=3.0,\; S=0.0,\; T=5.0,\; P=1.0; \; \gamma^\star = 3/7$\\
Harsh prisoner's dilemma & $R=3.0,\; S=0.0,\; T=6.0,\; P=2.0; \; \gamma^\star = 5/7$ \\
\bottomrule
\end{tabular}
\end{table}

\subsection{Experiment 3: PPNE verification}
\label{app:exp_3_ppne_verification}

This experiment uses the same base game, parameter domain, and projected gradient-ascent procedure as \zcref[S]{app:exp1}, with identical learning-rate, initialization scheme, and optimization horizon and other hyperparameters.

The goal of this experiment is not to study the learning trajectories themselves, but to verify whether a boundary candidate $\theta^\star$ is in fact a parametric program Nash equilibrium. By \zcref[S]{thm:boundary_characterization}, this reduces to a one-dimensional unilateral-deviation check for each player. Concretely, after selecting a candidate boundary point $\theta^\star$, we freeze the opponent at $\theta^\star_{-i}$ and evaluate the best-response objective $\theta_i \mapsto U_i(\theta_i,\theta^\star_{-i})$ over the full interval $[-5,5]$. In practice, this is implemented as a dense grid sweep with 1000 evaluation points, together with the corresponding derivative curve $\partial U_i/\partial \theta_i$. A candidate $\theta^\star$ is certified as a PPNE only if, for both players, the maximum of the swept best-response curve is attained at $\theta_i^\star$. Equivalently, the unilateral-improvement gap $\max_{\theta_i \in [-5,5]} U_i(\theta_i,\theta^\star_{-i}) - U_i(\theta^\star)$ must be numerically zero up to tolerance.

For the open-source condition, we use sigmoid semantics with coupling $\gamma=1.5$ as in \zcref[S]{fig:ppne_verification} where the best-response payoff is maximized at the cooperative boundary $\theta_i^\star=B$, so $(B,B)$ is a PPNE. In \zcref[S]{fig:ppne_verification_nonppne} we use $\gamma=0.7$, the payoff attains a strictly larger value at an interior deviation, showing that $(B,B)$ is not a PPNE despite being a boundary candidate. Candidate boundary points are taken to be $(5,5)$ in the open-source case and $(-5,-5)$ in the closed-source case. To compare with the gradient-ascent dynamics, we also run the same solver as in \zcref[S]{app:exp1} from 10 random initializations drawn uniformly from $[-2.5,2.5]^2$, using learning rate $0.15$, $600$ update steps, clipping bound $B=5$, and random seed $0$, which is the same configuration as in \zcref[S]{tab:exp1_hparams}.

The numerical verification results for $\gamma = 0.7$, intended not to be a PPNE, are shown in \zcref[S]{fig:ppne_verification_nonppne}. In this configuration, the open-source boundary candidate $(5,5)$ is not a PPNE, as for each player, the best-response curve achieves a strictly larger value away from the boundary, with a unilateral-improvement gap of approximately $1$, in contrast to \zcref[S]{fig:ppne_verification}, where the gap is approximately 0. This illustrates the distinction emphasized in the main text, i.e., the convergence of gradient dynamics toward a boundary region does not by itself imply that the limiting boundary point is an equilibrium of the induced game. The additional figure is included to show how the same sweep procedure certifies failure of the PPNE property when a profitable unilateral deviation exists.

\begin{figure}[htpb]
    \centering
    \includegraphics[width=\linewidth]{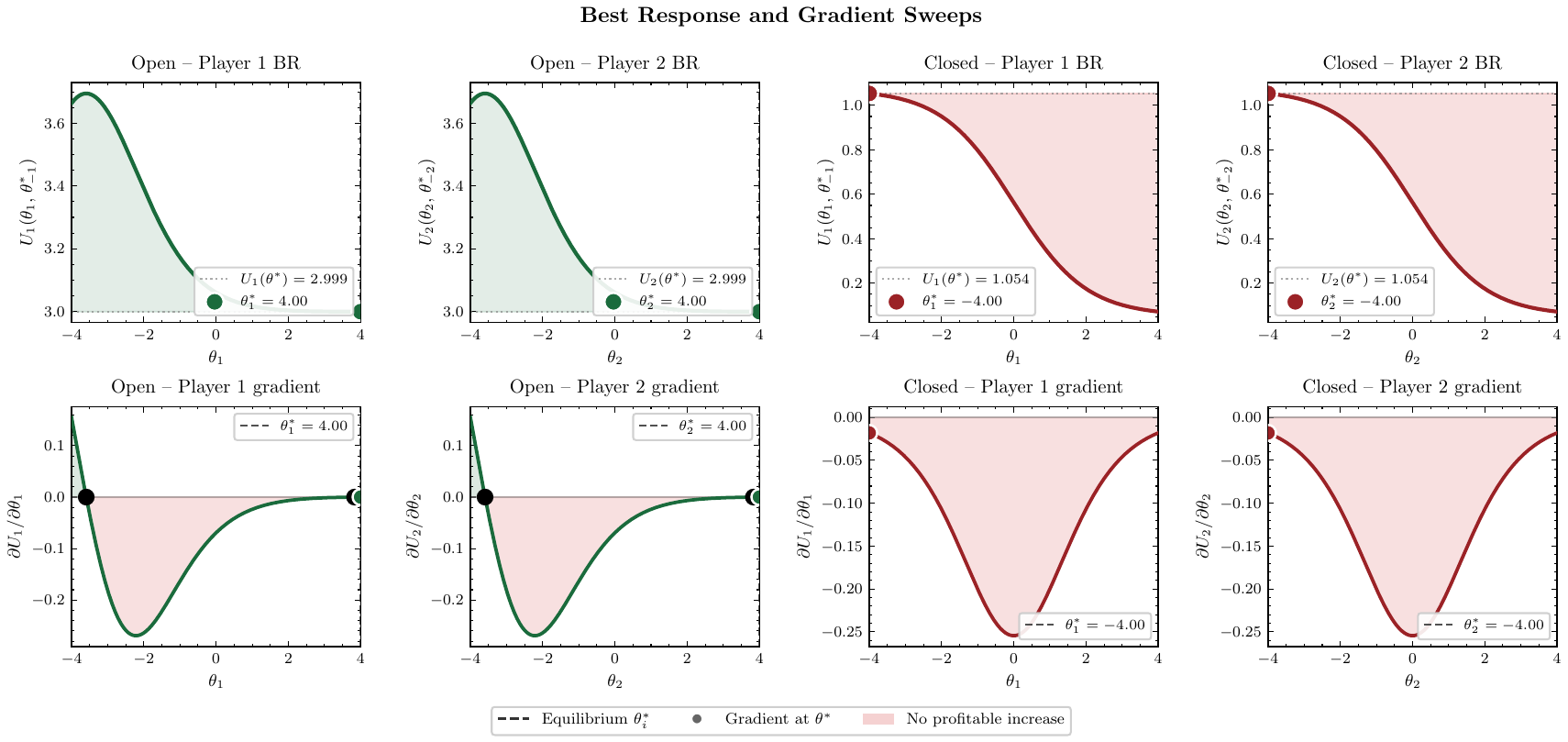}
    \caption{Boundary best-response verification for a non-PPNE open-source candidate. Using the same sweep procedure as in \zcref[S]{fig:ppne_verification}, but with weaker open-source coupling $\gamma=0.7$ rather than $\gamma=1.5$, the candidate boundary point $(B,B)$ fails the PPNE test. The best-response curves are not maximized at $\theta_i^\star=B$, so each player has a profitable unilateral deviation when the opponent is fixed at the boundary. This contrasts with the main-text setting in \zcref[S]{fig:ppne_verification}, where $\gamma=1.5$ and the cooperative boundary is verified as a true PPNE. The solid black dot in Panels (e) and (f) indicates the best response found by the fine-grid sweep.}
    \label{fig:ppne_verification_nonppne}
\end{figure}

\subsection{Experiment 4: Neural semantics}\label{app:exp_4_neural}

This experiment extends the gradient-ascent setup of \zcref[S]{app:exp1} to the neural semantics class introduced in \zcref[S]{sec:neural_semantics}. We evaluate the four canonical symmetric $2\times 2$ games shown in \zcref[S]{fig:neural_semantics_learning}, namely Stag Hunt, Chicken, Prisoner's Dilemma, and Harsh PD. In all conditions, we report the social welfare $U_1+U_2$ as a function of optimization time and average across multiple random seeds.

We compare six conditions. The first two are the sigmoid baselines from the main text, namely sigmoid-closed with $\gamma=0$ and sigmoid-open with $\gamma=1.5$. The next two use the neural semantics of \zcref[S]{sec:neural_semantics} with fixed residual weights. In the closed condition, we set $(\alpha,\beta)=(1,0)$, while in the open condition we set $(\alpha,\beta)=(1,1.5)$, so that $\beta/\alpha=1.5$ matches the open sigmoid coupling. In both cases, the network weights $(W_1,b_1,W_2)$ are sampled once at initialization and then held fixed throughout optimization. The same sampled residual network is used for the fixed closed and fixed open conditions, so the comparison isolates the effect of the linear coefficients.

The final two conditions jointly learn both the player parameters and the semantics parameters. In the warm-start condition, the learnable parameters are initialized with $(\alpha_0,\beta_0)=(1,1.5)$, so the initial first-order ratio already lies in the cooperative regime predicted by \zcref[S]{app:neural_threshold}. In the cold-start condition, they are initialized with $(\alpha_0,\beta_0)=(0,0)$. For each player, the learnable parameter vector consists of the current scalar parameter $\theta_i$, the linear coefficients $\alpha,\beta$, and the flattened residual-network weights. Gradients are approximated by finite differences, and all parameters are updated by simultaneous gradient ascent on each player's own payoff.

Player parameters are initialized from a clipped normal distribution centered at zero, with a clipping radius of $2.0$. During optimization, the player parameters are clipped to $[-10,10]$, while the learnable semantics parameters are clipped componentwise to $[-5,5]$. The residual network uses hidden width $K=4$, initialization scale $0.05$, and zero initial bias vector $b_1$. The full set of hyperparameters is reported in \zcref[S]{tab:neural_semantics_hparams}.

\begin{table}[t]
\centering
\caption{Hyperparameters for the neural-semantics experiment in \zcref[S]{fig:neural_semantics_learning}. The fixed-neural conditions use the same sampled residual network in both the closed and open variants, differing only in the linear coefficients $(\alpha,\beta)$. The learned neural conditions jointly optimize the player parameters and all semantics parameters.}
\label{tab:neural_semantics_hparams}
\begin{tabular}{ll}
\toprule
Hyperparameter & Value \\
\midrule
Open sigmoid coupling $\gamma$ & $1.5$ \\
Hidden width $K$ & $4$ \\
Residual weight scale & $0.05$ \\
Residual-network seed & $0$ \\
Finite-difference step $\varepsilon$ & $10^{-5}$ \\
Player learning rate & $0.10$ \\
Semantics-parameter learning rate & $0.005$ \\
Player clipping bound & $10.0$ \\
Semantics-parameter clipping bound & $5.0$ \\
Player initialization clip & $2.0$ \\
Number of steps for fixed-semantics runs & $1000$ \\
Number of steps for learned-semantics runs & $1000$ \\
Number of seeds for fixed-semantics runs & $5$ \\
Number of seeds for learned-semantics runs & $5$ \\
\bottomrule
\end{tabular}
\end{table}

The main empirical pattern is that the neural class reproduces the sigmoid baselines when the first-order ratio $\beta/\alpha$ is matched, but joint optimization succeeds reliably only under warm initialization. This is consistent with \zcref[S]{app:neural_threshold}, which identifies $\beta/\alpha>\gamma^\star$ as the condition under which cooperation is locally attractive at the symmetric midpoint. The fixed-open and warm-start learned conditions satisfy this requirement from the outset, whereas the cold-start learned condition must discover it through optimization and fails to do so.
\clearpage
\newpage
\section{Camera-Ready Edits}
We improved the quality of the figures, bolded and shortened their titles, and improved the captions within the manuscript. We also added an Acknowledgment section before the references.

\end{document}